\documentclass{llncs}
\usepackage{amsfonts,amsmath,amsfonts,amsmath,amssymb}
\usepackage{color}

\usepackage[noend]{algorithmic}
\usepackage{algorithm}
\usepackage{graphicx}

\usepackage{enumerate}
\usepackage{esvect}

\newenvironment{proofof}[1]{\begin{proof}[of {#1}]}{\end{proof}}

\newcommand*\samethanks[1][\value{footnote}]{\footnotemark[#1]}

\newcommand{\eat}[1]{}

\newcounter{proccnt}

\begin{document}\sloppy

\title{Low-Risk Mechanisms for the Kidney Exchange Game}
\author{Hossein Efsandiari\inst{1}\thanks{Supported by NSF grant number 1218620.}, Guy Kortsarz\inst{2}\samethanks[1]}

\institute{University of Maryland
	\and
	Rutgers University  at Camden}

 \maketitle

\begin{abstract}
In this paper we consider the \emph{pairwise kidney exchange game}. This game naturally appears in situations that some service providers benefit
from pairwise allocations on a network, such as the kidney exchanges between hospitals. 

Ashlagi et al. \cite{ashlagi2013mix} present a $2$-approximation randomized truthful mechanism for this problem. This is the best known result in this setting with multiple players. However, we note that the variance 
of the utility of an agent in
this mechanism may be as large as
$\Omega(n^2)$, which is not desirable in a real application. In this paper we resolve this issue by 
providing a $2$-approximation randomized truthful mechanism in which the variance 
of the utility of each agent is at most $2+\epsilon$.

Interestingly, we could apply our technique to design a {\em deterministic} mechanism such
that, if an agent deviates from the mechanism, she does not gain more than 
$2\lceil \log_2 m\rceil$.
We call such a mechanism an {\em almost truthful} mechanism.
Indeed, in a practical scenario, an almost truthful mechanism is likely to imply a truthful mechanism.
We believe that our approach can be used to design low risk or almost truthful mechanisms for other problems.
\end{abstract}

\newpage

\section{Introduction}
Kidney transplant is the only treatment for several types of kidney 
diseases. Since people have two kidneys and can survive with only one kidney, 
they can potentially donate one of their kidneys. 
It may be the case that a patient finds a family member or a friend 
willing to donate her kidney. Nevertheless, 
at times the kidney's donor is not compatible with the patient. 
These patient-donor pairs create a list of incompatible pairs.
Consider two incompatible patient-donor pairs. 
If the donor of the first pair is compatible with the patient of the second pair and vise-versa, we can efficiently serve both patients without affecting the donors.

In this paper we consider pairwise kidney exchange, even though
there can be a more complex 
 combinations of transplantation of kidneys,
 that involves three or more  pairs. Nevertheless, such chains are complicated to deal with in the real life applications  since they need six or more simultaneous surgeries. 

To make the pool of donor-patient pairs larger,
hospitals combine their lists of pairs to one big pool,
trying to increase the number of treated 
patients by exchanging pairs from different hospitals.
This process is managed by some national supervisor.
A centralized mechanism can look at all of the hospitals 
together and increase the total number of kidney exchanges. 
The problem is that for a hospital its key interest is to increase the number of its own served patients. 
Thus, the hospital may not report some  patient-donors 
pairs, namely, the hospital may report a partial list.
This partial list is then matched 
by the national supervisors. Undisclosed 
set of pairs are matched by the hospitals locally, 
without the knowledge of the supervisor.
This may have a negative effect on 
the number of served patients.

 A challenging problem is to design a mechanism for the national supervisor, 
to convince the hospitals not to hide information,
and report all of their pairs. 
In fact, if hiding any subset of vertices does not increases the utility of an agent, she has no intention to hide any vertex.
Moreover, in a real application, hiding vertices involves extracting the information about other agents and finding the right subset of vertices to hide, which is costly itself. Thus, in a real application, if the loss stemming from being truthful is {\em negligible}, it is likely that the hospital will absorb the small loss and remain truthful. In this work we do not define any cost for deviating form the mechanism. However, we seek to find a mechanism such that the gain by not being truthful is \emph{exponentially smaller} than the size of the problem.

\subsection{Notations and Definitions}
To model this and similar situations hospitals are called {\em agents}, and each patient-donor pair is modeled by a vertex.
Let $m$ be the number of agents. Each agent owns a disjoint set of vertices. We denote the vertex set of the $i$-th agent by $V_i$ and $\vv V = \{V_1,V_2,...,V_m\}$ is called 
the vector of vertices of the agents. Denote an instance of the kidney exchange problem by $(G,\vv V)$, where $G$ is the underlying graph and $\vv V$ is the vector of vertices of the agents.
Each vertex in $G=(V,E)$ belongs to exactly one agent. Thus, $V=\cup_{i=1}^{m}V_i$ holds.

In this game, the utility of an agent $i$ is the expected number of matched
vertices in $V_i$ and is denoted by $u_i$. 
Similarly, the utility of an agent $i$ with respect to a matching $M$ 
is the number of vertices of $V_i$ 
matched by $M$ and is denoted by 
$u_i(M)$. The social welfare of a mechanism is the size of the output matching.

A mechanism for the kidney-exchange game is the mechanism employed by the national supervisor
to choose edges among the reported vertices. 
The process is a three step process. 
First the agent expose some of their vertices. 
Then the mechanism chooses a matching on the reported graph.
Finally, each agent matches her unmatched vertices, including her non disclosed vertices, privately.

Formally, 
a kidney exchange mechanism $F$ is a function from an instance of a kidney exchange problem $(G,\vv V)$ to a matching $M$ of $G$. The mechanism $F$ may be randomized. We say a kidney exchange mechanism is truthful if no agent has incentive to hide any vertex i.e., for each agent $i$, we have 
\begin{align*}
&\forall_{V'_i \in V_i} &u_i(F(G)) \geq u_i(F(G\backslash V'_i)) + u_i(F(G\backslash V'_i), V'_i)
\end{align*}
where $u_i(F(G \backslash V'_i), V'_i)$ is the [expected] number of vertices that agent $i$ matches privately if she hides $V'_i$. We define \emph{almost truthful} mechanisms as follow. 
\begin{definition}
	We say that a mechanism is almost truthful
	if by deviating from the mechanism, an agent can gain at most 
	an additive factor of $O(\log m)$
	vertices in the revenue, with $m$, the number of agents.
\end{definition}

Consider that in a real application finding the right subset of vertices to hide is costly.
Indeed, this cost involves extracting the information of $m$ other agents. Thus, we hope that in an almost truthful mechanism agents reports the true information. Consider that, if we let the gain of deviating from the mechanism to be a constant fraction of the revenue of the agent, this gain may become considerable for agents with very large revenue. Thus, in a real application, the agents with very large revenue may prefer to form a team to find the right set to hide.

Remark that, in this paper we do not consider a cost for deviating from the mechanism, and thus, we use two different notations of truthful mechanisms and almost truthful mechanisms.

Given that 
some pairs are undisclosed,
we say a kidney exchange mechanism $F$ is $\alpha$-approximation if for every graph $G$ the number of matched vertices in the maximum matching of $G$ is at most $\alpha$ times the expected number of matched vertices in $F(G)$. This means that for every graph $G$ 
\begin{align*}
\frac {|Opt(G)|}{E[|F(G)|] }\leq \alpha,
\end{align*} 
where $Opt(G)$ is the maximum matching in graph $G$, and the expectation is over the run of the mechanism $F$.


We define the notion of \emph{bounded-risk} mechanisms as follow.
\begin{definition}
A mechanism is a bounded-risk mechanism if the variance of the utility of each agent is is bounded by a constant.
\end{definition}

\subsection{Related Work}
The model considered in this paper was initiated by 
S\"onmez and \"Unver \cite{sonmez2013market} and Ashlagi and Roth \cite{ashlagi2011individual}.
S\"onmez and \"Unver \cite{sonmez2013market} show that 
there is no
deterministic 
truthful mechanism  
that gets the maximum possible social welfare.
See Figure \ref{fig:hard}. In this example,
the number of vertices is odd. Therefore, any mechanism that provide a maximum matching leaves exactly one vertex unmatched. Consider a mechanism that leaves a vertex of 
the first agent unmatched. In this case the utility of the first agent is $2$. 
If this agent hides the fifth and the sixth vertices, 
any maximum matching matches the first vertex to the second vertex and the third vertex to the fourth vertex. Later, agent one matches the fifth and sixth vertices, privately. 
This increases the utility of the first agent to $3$, and means that such a mechanism is not truthful. 
Similarly, if the mechanism leaves a vertex of the second agent unmatched, she can increases her utility by hiding the second and third vertices and matching them privately. This shows that a mechanism that always reports a maximum matching is not truthful.

\begin{figure}
\begin{center}
\includegraphics[width=8cm]{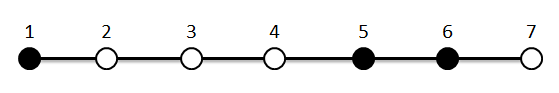}
\caption{Black vertices belong to the first agent and white vertices belong to the second agent.}
\label{fig:hard}
\end{center}
\end{figure}

Achieving social welfare optimal  mechanisms, which are truthful,
is thus not possible. 
However, achieving approximate truthful mechanisms may be possible.
Ashlagi et al. \cite{ashlagi2013mix} used the same example as in Figure \ref{fig:hard} 
to show that there is no deterministic truthful mechanism 
for the kidney-exchange game, with approximation ratio better than $2$. Moreover, they show that there is no 
randomized truthful mechanism with an approximation ratio better than $8/7$. 
They also introduce a deterministic $2$-approximation truthful mechanism for the two player kidney exchange game and a randomize $2$-approximation truthful mechanism for the multi-agent kidney exchange game. Later Caragiannis et al. \cite{caragiannis2011improved} improved the approximation ratio 
for two agents to an expected $3/2$-approximation truthful mechanism.
It is conjectured  that there is no deterministic constant-approximation truthful mechanism for the multi-agent kidney exchange game, even for three agents \cite{ashlagi2013mix}. 

Almost truthful mechanisms has been widely studied (See \cite{dughmi2011approximately}, \cite{kothari2005approximately} and \cite{lesca2012almost}) with slightly different definitions. However, all use the concept that an agent should not gain more than small amount by deviating from the truthful mechanism.

\subsection{Our Results}
First, we show that the variance of the utility of an agent in the mechanism proposed by Ashlagi et al. \cite{ashlagi2013mix} may be as large as  $\Omega (n^2)$, where $n$ is the number of vertices. The variance of the utility can be interpreted as the risk of the agent caused by the randomness in the mechanism. Indeed, in a real application agents prefer to take less risk for the same expected utility. In Section \ref{sec:rand}, we provide a tool to lower the variance of the utility of each agent in a kidney exchange mechanism while keeping the expected utility of each agent the same. The following theorem 
is an application of this tool to the mechanism proposed by Ashlagi et. al. \cite{ashlagi2013mix}.
low variance.

\begin{theorem}
\label{1}
There exists a bounded-risk truthful $2$-approximation 
mechanism for multi-agent kidney exchange. Specifically, in this mechanism the variance of the utility of each agent is at most $2+\epsilon$, where $\epsilon$ is an arbitrary small constant.
\end{theorem}

Later, in Section \ref{sec:det}, we provide a derandomization of our mechanism. Specifically, we design an almost truthful deterministic $2$-approximation mechanism for this problem. To the best of our knowledge this is the first non-trivial deterministic mechanism for the multi-agent kidney exchange game.

\begin{theorem}
\label{2}
There exists an almost truthful deterministic $2$-approximation mechanism for multi-agent kidney exchange.
\end{theorem}


\section{A truthful mechanism with small utility variance}\label{sec:rand}

Ashlagi et al. in EC'10 \cite{ashlagi2013mix} study the
multi-agent kidney exchange game. They provide a polynomial time truthful $2$-approximation mechanism called Mix and Match. The 
Mix and Match mechanism is described as follows;
independently label each agent either by $1$ or $0$ 
each with probability $0.5$. Remove the edges between different agents with the same labels, i.e., for each edge $(u,v)\in E$, if $u$ and $v$ belongs to different agents and these agents have the same label, remove the edge $(u,v)$ from $G$. 
Let $G'$ be the new graph. 
Consider all matchings in $G'$ that contain a maximum matching over 
the induced subgraph of each agent separately. Output the one with the maximum cardinality.
Ties are broken serially in favor of agents with label $1$.
The following example shows that in this mechanism the variance of an agent utility may be as large as $\Omega (n^2)$.
\begin{example}\label{example:badMech}
Consider a game with three agents. Each agent has $\frac{n}{3}$ vertices, where $n$ is the number of vertices in the graph. There is a perfect matching with $\frac{n}{3}$ edges between vertices of agent $1$ and agent $2$ and there is no other edges (see Figure \ref{fig:example}). 
In this example, with probability $0.5$, agent $1$ and agent $2$ get the same label and all of the edges between these two agents are removed. 
In this case, all edges are removed and thus the utility of each agent is zero.  However, with probability $0.5$, agent $1$ and agent $2$ get different labels and we have a matching of size $\frac n 3$ between the vertices of these two agents. In this case, the utility of agent $1$ is $\frac n 3$. Therefore, the variance of the first agent utility is 
\begin{align*}
\sigma^2= 0.5(0-\frac n 6)^2+0.5(\frac n 3 -\frac n 6)^2 = \frac {n^2} {36},
\end{align*}
which is $\Omega (n^2)$.
\begin{figure}
\begin{center}
\includegraphics[width=7.8cm,height=6cm]{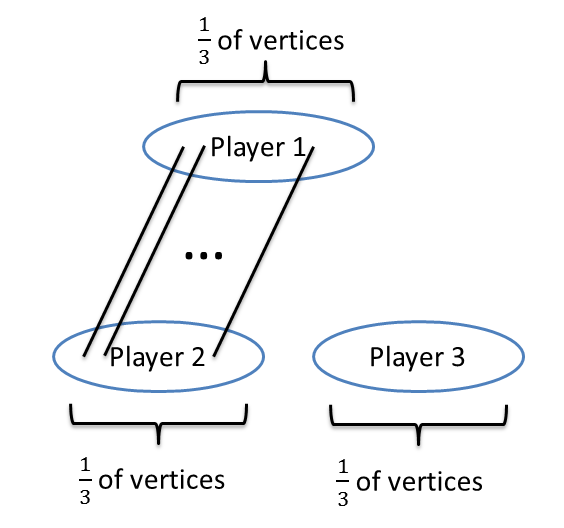}
\caption{In this example, the variance  of the first agent utility, is $\Omega (n^2)$}
\label{fig:example}
\end{center}
\end{figure}

\end{example}

%
%

Our mechanism uses a randomized truthful mechanism as a core mechanism. We take two matchings
resulting from two independent runs of the core mechanism.
These two matchings are combined into a new 
matching.
The way we choose the new matching 
is randomized too.
The new matching preserves the expected utility of each agent, and in addition, decreases the variance of the utility of each agent 
by a constant factor.
This gives a mechanism with a lower utility-variance. We repeat this procedure iteratively (See Figure \ref{fig:layers}) and decrease the variance of the utilities to $O(1)$. We show that for this purpose it is enough 
to apply the combination of two matchings mechanism a 
logarithmic number of times. For the purpose of this section, we use Mix and Match as the core mechanism and show that it gives us a truthful $2$-approximation 
mechanism such that the variance of each utility is at most $2+\epsilon$, where $\epsilon$ is an arbitrary small constant.

\begin{figure}
\begin{center}
\includegraphics[width=8cm]{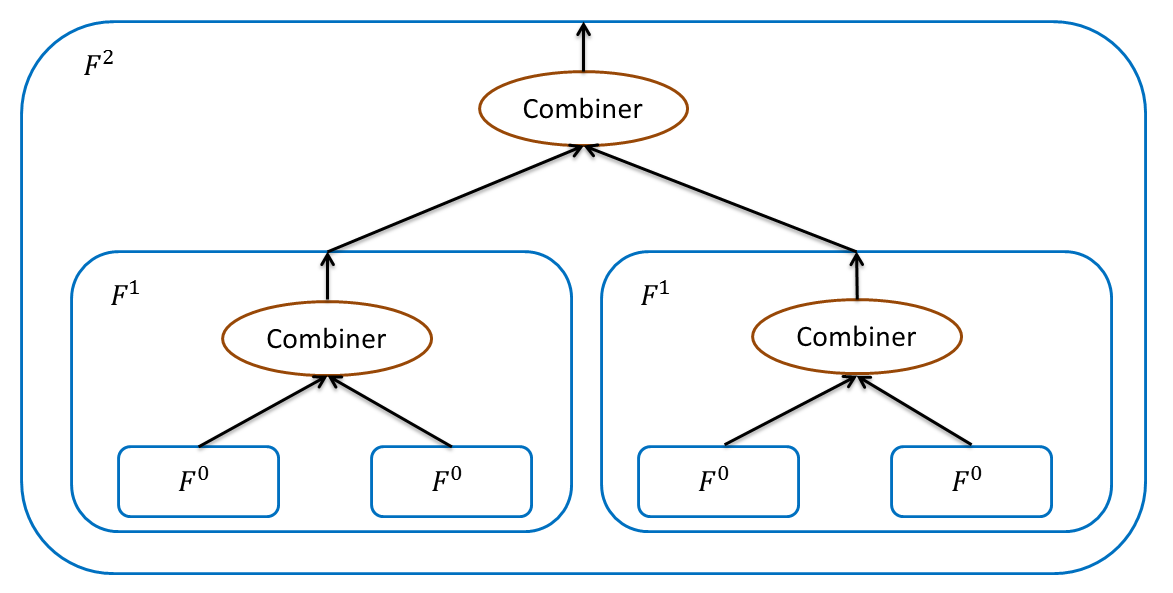}
\caption{Hierarchy of mechanisms}
\label{fig:layers}
\end{center}
\end{figure}

One can think of our mechanism as a multi-layered mechanism. The layer-$0$ mechanism is the core mechanism. In the $i$-th layer we combine two outputs of the layer $i-1$ mechanism. Lemma \ref{lm:main} shows that we can use the lower layer mechanism and create a mechanism where the variance of the utilities is almost halved. 

Note that the utility of an agent can be completely 
different for two matchings $M_1$ and $M_2$. 
Let $M_1\oplus M_2$ denote the symmetric difference of $M_1$ and $M_2$.
There may be a path from 
$u$ to $v$ in $M_1\oplus M_2$ in which $u$ is 
a vertex of agent $i$ and $v$ is a vertex of agent 
$j$ and $i\neq j$. One of the two matchings will have a utility
smaller by $1$ for agent $i$. 
As the number of such paths in $M_1\oplus M_2$ may be very large,
the difference in utility of an agent with respect to
$M_1$ and $M_2$ can be very large.
We show how to find two matchings $N_1$ and $N_2$ 
such that 
the utility of each agent with respect to these two matchings is almost equal.

Let $(G,\vv V)$ be an instance 
of the kidney exchange graph.
Consider two matchings $M_1$ and $M_2$ derived by independent 
runs of the previous layer. Let 
$P=\{p_1,p_2,\dots,p_k\}$ be a subset of distinct paths 
in 
$M_1\oplus M_2$ (ignoring cycles).
\begin{definition}
The {\em contraction graph} 
$Cont((G,\vv V),P)$ is defined as follows:
\begin{itemize}
\item Each vertex in $Cont((G,\vv V),P)$ corresponds to one agent in $(G,\vv V)$.
\item There is an edge in $Cont((G,\vv V),P)$ between agent $i$ and $j$ if and only if there is a path in $P$ that begins with a vertex of agent $i$ and ends with a vertex of agent $j$.
\end{itemize}
We call this graph {\em the contraction graph} because 
the paths are replaced by edges.
When the instance of the kidney exchange game is clear from the context, we drop $(G,\vv V)$ 
from the notation of the contraction graph.
\end{definition}

The following lemma proves that any two matchings can be 
transformed into two other matchings such that
for every agent $i$ the utility of agent $i$ in the two new matchings has a difference of at most $2$.
\begin{lemma} \label{lm:str}
Let $M_1$ and $M_2$ be two matchings of graph $G$. There exist two matchings $N_1$ and $N_2$ such that for any agent $i$ we have
\begin{itemize}
\item $|u_i(N_1)-u_i(N_2)|\leq 2$ and
\item $u_i(N_1)+u_i(N_2)=u_i(M_1)+u_i(M_2)$.
\end{itemize}
\end{lemma}
\begin{proof}
We decompose $M_1 \oplus M_2$ into two different matchings $N'_1$ and $N'_2$ 
such that  $N_1'\cup N_2'=M_1\oplus M_2$. Then define 
$N_1$ and $N_2$ as 
 $N_1=N'_1\cup (M_1 \cap M_2)$ and $N_2=N'_2\cup (M_1 \cap M_2)$ respectively. 
Clearly, an edge $e$ 
 belongs to exactly one of $N_1$ or $N_2$,  if and only if $e$ belongs to exactly one of  $M_1$ or $M_2$. In addition, $e$ belongs to 
{\em both} of $N_1$ and $N_2$ if and only if it belongs to {\em both} of $M_1$ and $M_2$. This means that the equality $u_i(N_1)+u_i(N_2)=u_i(M_1)+u_i(M_2)$ holds for all agents. 
It holds true, regardless of the way we decompose $M_1\oplus M_2$ into $N'_1$ and $N'_2$. 
We now describe our approach to achieve the main property, namely, 
change  $M_1\oplus M_2$
into two matchings $N'_1$ and $N'_2$ such that  $N'_1\cup N'_2=M_1\oplus M_2$ and 
$|u_i(N'_1)-u_i(N'_2)|\leq 2$.

Consider the subgraph induced in $G$ by $M_1 \oplus M_2$. 
The degrees of vertices in this graph are either zero or one or two. 
There are three types of connected components:
cycles, even-length paths and odd-length paths. 
We explain how to decompose 
the different parts of  $M_1 \oplus M_2$  in these three cases.
Every path  $p$ decomposes into two matchings 
$M^p_1$ and $M^p_2$. In any such decompositions, all the vertices of $p$ except the endpoints are covered by both $M^p_1$ and $M^p_2$. 
\begin{itemize}
\item \textbf{Case 1: Components that are cycles.} 
Each of these cycles is the union of two matchings. It means that every other edge in each cycle belongs to one of the matchings. We add one of these matchings to $N'_1$ and we add the other one to $N'_2$. Since these two matchings cover the same set of vertices, they have the same effect on the utility of agents.
\item \textbf{Case 2: Edges of even length paths.} Let $p$ be an even size 
path between two vertices $v$ and $u$. Let  $M^p_1$ and $M^p_2$ 
be a decomposition of $p$ into matching such that 
 $M^p_1$ covers all vertices in $p$ except $u$ and $M^p_2$ covers all vertices in $p$ except $v$. For each path, we add one of $M^p_1$ and $M^p_2$ to $N'_1$ and add the other one to $N'_2$. However, the assignment cannot be 
arbitrary.
The assignment is derived by performing 
computations on the contraction graph.
We represent the selection of $M^p_1$ by directing the edge 
in the contraction graph from the agent that contains $u$ to the agent that contains $v$. Note that we deal with all even paths {\em  simultaneously}.
The difference of the outgoing and in-going degrees of each agent exactly equals the difference of her utilities caused by edges of even length paths in $N'_1$ and $N'_2$. Thus, we just need to direct edges of the contraction graph of even length paths, 
so as to minimize this difference. 
We can direct the edges of this graph such that for each vertex the difference of outgoing and in-going edges is at most one. 
This is done by adding a matching between the odd degree vertices, directing the edges through an Eulerian cycle and removing the added edges.
We adopt this strategy to get our almost balanced in and out degree
pair of matchings, derived from all even sized paths in $M_1\oplus M_2$.
\item \textbf{Case 3: Odd length paths.} 
Let $p$ be a path between vertices $v$ and $u$ which has an odd number of edges. We can decompose $p$ into two matchings $M^p_1$ and $M^p_2$ such that $M^p_1$ covers all vertices in $p$ and $M^p_2$ covers all vertices in $p$ except $v$ and $u$. Again in this case, for each path, we add one of $M^p_1$ or $M^p_2$ to $N'_1$ and add the other one to $N'_2$. 

In one of $N'_1$ and $N'_2$,  
both endpoints are matched and in the other none of the two endpoints are matched.

We represent the selection of $M^p_1$ by coloring the edge corresponding to $p$ blue. Otherwise, we color the edge red. 
Let the blue (red) degree of a vertex be the number of blue (red) edges 
touching the vertex.
The difference between the red and blue degrees of each agent,
 exactly equals  
the difference in her utilities caused by odd length paths in $M_1\oplus M_2$.

We can color the edges of any arbitrary graph with blue and red such that for each vertex the difference between the red and blue degree is at most $2$.
This again is done by adding a matching of dummy edges 
between the odd degree vertices,
and coloring every second edge in 
the Eulerian cycle red and every other second edge blue.
Then we remove the fake edges added at the beginning.
Note that the start vertex of the cycle may be touched by two 
red or two blue, edges. 
On the other hand, the other vertices in this cycle 
are touched by the same number of blue and red edges.
We use the following rule: if the start vertex has a dummy edge,
we use this edge as the first in the Euler cycle.
This is done such that the difference of blue and red degrees 
will not accumulate to $3$ ($2$ may be added to the difference 
due to the fact that
the path starts and ends in the same color, and an additional $1$ can be added to the difference when we take the 
dummy edge out).
This clearly implies a difference of $2$ in the utility of any agent 
with respect to the new matchings.
\end{itemize}

If we combine the matchings of case $2$ and case $3$, 
the difference of the utilities for any agent 
with respect to  $N'_1$ and $N'_2$
may grow up to at most $3$. We want to avoid this situation.
Let $i$ be some agent for which the utility difference is $3$. 
The difference of the utilities is derived as follows:
\begin{itemize}
\item  Agent $i$ has a difference of two between the number of red edges and blue edges. Note that this means that the other agents 
have a difference of at most $1$ between the number 
of red and blue edges,
as the the agents are not start vertices of the cycle.
\item  Agent $i$ has a difference of one between the number of outgoing and in-going edges. This means that the vertex is an odd degree vertex.
\item The effect of these two differences accumulate and cause a difference of three between $N'_1$ and $N'_2$.
\end{itemize}

In this case, we flip the color of edges in the component that contains $i$. 
This decreases {\em for i} the difference of $N'_1$ and $N'_2$ 
from $3$ to $1$. Note that any vertex that was not a start 
vertex of the Euler cycle has a difference of at most $1$ in the edge
coloring stage.
The flipping of colors still
implies that the maximum difference in the utility of every
agent that is not a beginning of a cycle, is at most $1$.
Together with the difference caused by even length paths,
the total difference is at most $2$. 
Note that cycles of two different agents with difference $2$ in
their  utility,
are disjoint.
Only one agent with
two red or two blue edges can exist in every connected component. \qed
\end{proof}

The following lemma uses lemma \ref{lm:str} to combine outcomes of two independent runs of a mechanism. 
\begin{lemma}\label{lm:main}
Let $F$ be a mechanism for the multi-agent kidney exchange game and let $x_i$ be the random variable of the utility of agent $i$ in the  mechanism $F$. Then there exist a mechanism $F^*$ such that for every input graph and every agent $i$ the following holds:
\begin{align*}
Var(y_i)\leq \frac{Var(x_i)} 2 +1,
E[y_i]= E[x_i],
\end{align*}
where $y_i$ is the random variable that indicates the utility of agent $i$ in mechanism $F^*$.
\end{lemma}
\begin{proof}
We run mechanism $F$ two times independently. Let $M_1$ and $M_2$ be the random matchings resulting from these two runs. We apply Lemma \ref{lm:str} on  
$M_1$ and $M_2$ and let 
$N_1$ and $N_2$ be the resulting matchings. 
The mechanism $F^*$ chooses one of the two matchings $N_1$ and $N_2$ uniformly at random. 
Note that:
\begin{align*}
E[y_i]=\frac{E[u_i(N_1)]+E[u_i(N_2)]}{2}=\frac{E[u_i(M_1)]+E[u_i(M_2)]}{2}= \frac {2E[x_i]}{2}=E[x_i] 
\end{align*}
where the second equality is an application of Lemma \ref{lm:str}. This means that each agent has the same expected utility in $F$ and $F^*$. 
Now, we need to bound the variance of the utilities of the agents in $F^*$. Let $D_i$ be the difference between $u_i(N_1)$ and the average of $u_i(M_1)$ and $u_i(M_2)$. Note that $u_i(N_1)=\frac {u_i(M_1)+u_i(M_2)}2+D_i$ and $u_i(N_2)=\frac {u_i(M_1)+u_i(M_2)}2-D_i$. Thus, we have
\begin{align*}
Var\left (u_i(N_1)\right ) =& Var\left (\frac {u_i(M_1)+u_i(M_2)}2+D_i\right ) \\=& \frac{u_i(M_1)}{4}+\frac{u_i(M_2)}{4}+Var(D_i)+Cov\left (\frac {u_i(M_1)+u_i(M_2)}2,D_i\right )\\=& \frac{Var(x_i)}{2}+Var(D_i)+Cov\left (\frac {u_i(M_1)+u_i(M_2)}2,D_i\right ).
\end{align*}
We get $Var(u_i(N_2)) = \frac{Var(x_i)}{2}+Var(D_i)-Cov\left (\frac {u_i(M_1)+u_i(M_2)}2,D_i\right )$ in a similar way. 
The variance of $y_i$ is the average of variances of $u_i(N_1)$ and $u_i(N_2)$. Thus, 
\begin{align}\label{lm:eq:var}
Var(y_i)\leq \frac{Var(x_i)} 2 +Var(D_i).
\end{align}
It remains to bound $Var(D_i)$. From Bhatia-Davis Inequality \cite{bhatia00} for any random variable $X$:
$Var(X) \leq (Sup(X) - \mu)(\mu - Inf(X))$, where $\mu$ is the expected value of $X$, $Sup(X)$ is the supremum of $X$ and $Inf(X)$ is the infimum of $X$. By applying Bhatia-Davis Inequality to the random variable $D_i$, we have
\begin{align*}
Var(D_i) &\leq (Sup(D_i) - \mu)(\mu - Inf(D_i)) \leq (1 - \mu)(\mu+1)  \leq 1
\end{align*}
where the second inequality is by definition of $D_i$ and Lemma \ref{lm:eq:var}.
Combining this with inequality (\ref{lm:eq:var}), gives us 
 $Var(y_i)\leq \frac{Var(x_i)} 2 +1$ as desired. \qed
\end{proof}

Lemma \ref{lm:main} provides a way to decrease the variance of the utilities,
iteratively. New we are ready to prove Theorem \ref{1}.

\begin{proofof}{Theorem \ref{1}}

Let $F^0$ be the Mix and Match mechanism and let $F^{i}$ be the combination of two independent runs of $F^{i-1}$ from Lemma \ref{lm:main}. 
We call the mechanism {\em a multi-layered}
mechanism and  $F^i$ denoted $i$-th layer mechanism in the multi-layered 
mechanism. 
It is easy to see that $F^k$ is a combination of $2^k$ independent runs of the 
Mix and Match mechanism which is the layer $0$ mechanism. 
Recall that all those combinations preserve the expected utility of every agent. 
Thus, this process  preserves the social welfare function, which is the sum of utilities of all of the agents. Thus, the assumption that $F^0$ is a $2$-approximation mechanism immediately gives us that $F^k$ is a $2$-approximation mechanism, for any $k$.

Now, we show that for any $k$, $F^k$ is truthful. We prove this by contradiction.
Suppose that $F^k$ is not truthful. Without loss of generality we  assume that if agent $1$ deviates, her expected utility increases. 
Since, $F^k$ preserves the expected utility of each agent, this deviation should increase the expected utility of this agent all the way back to $F^0$. However,  this contradicts the truthfulness of $F^0$. Therefore, $F^k$ is truthful.

Now, we need to bound the variance of the utility of each agent. Without loss of generality, we fix an agent $i$ and bound the variance of the utility 
of that agent. 
The same bound holds for all agents, by symmetry. 

Let $\sigma^2$  be the variance of the utility of the agent in $F^0$. We prove by induction that the variance of the utility of this agent in $F^k$ is $\frac {\sigma^2}{2^k}+2-\frac 2 {2^k}$. The base case is clear since $\frac {\sigma^2}{2^0}+2-\frac 2 {2^0}=\sigma^2$. Assume the bound for  $F^k$
and then we prove it for $F^{k+1}$. 
By applying Lemma \ref{lm:main}, for $F^{k+1}$ the variance of the utility of the agent is at most
\begin{align*}
\frac{\frac {\sigma^2}{2^k}+2-\frac 2 {2^k}}{2}+1 = \frac {\sigma^2}{2^{k+1}}+1-\frac 1 {2^k} + 1=  \frac {\sigma^2}{2^{k+1}}+2 -\frac 2 {2^{k+1}}.
\end{align*}
This completes the induction. The utility of an agent cannot exceed the total number of vertices. Thus, we have $\sigma^2\leq n^2$. If we set $k$ to $2\log(n)+log(\frac 1 {\epsilon})$, then the variance of the utility of each agent in $F^k$ is at most
\begin{align*}
\frac {\sigma^2}{2^k}+2-\frac 2 {2^k} &= \frac {\sigma^2}{2^{2log(n)+log(\frac 1 {\epsilon})}}+2-\frac 2 {2^{2log(n)+log(\frac 1 {\epsilon})}}\\&=\frac {\sigma^2\epsilon}{n^2}+2-\frac {2\epsilon} {n^2}\leq \epsilon+2-\frac {2\epsilon } {n^2} \leq 2+\epsilon.
\end{align*}
We note that the running time is polynomial.
Indeed,  
running  $F^k$, combines $2^{k-i}$ instances of $F^0$, mechanisms. 
Since we set $k$ to $2\log n+log(\frac 1 {\epsilon})$, $F^k$ operates on  $\frac {n^2}{\epsilon}$ instances of $F^0$ and contains $\frac {n^2}{\epsilon}-1$ combinations of matchings in higher levels.
Since here both $F^0$ and the combination process runs in polynomial in $n$,
$F^k$ runs polynomial time as well.\qed
\end{proofof}

\section{An almost truthful deterministic mechanism}\label{sec:det}
In some applications, agents may not accept any risk. In this section, we modify our randomized mechanism to a deterministic one. This deterministic mechanism is not truthful anymore. However, it is almost truthful i.e.,
by deviating from the mechanism, an agent may gain an additive factor of at most 
$2 \lceil \log_2(m)\rceil$, where $m$ is the number of agents. 


%
%

The analysis of the Mix and Match mechanism does not use
the property that the labels of agents are 
fully independent. 
It just uses the fact that for every two fixed agents $i$ and $j$, with probability $0.5$, we assign different labels to agents $i$ and $j$.
In fact, this holds even if we use $m$ {\em pairwise 
independent} random bits. 
We can generate $m$ pairwise independent random bits using $\lceil \log_2(m) \rceil$ 
fully independent random bits \cite{luby1986simple}. We call this modified mechanism that just uses $\lceil \log_2(m) \rceil$ random bits, Modified Mix and Match.

\begin{proofof}{Theorem \ref{2}}
For simplicity of notation, we replace $\log_2(.)$ by $\log(.)$. Our deterministic mechanism can be seen as a multi-layered 
mechanism defined as follows. 
In Layer $0$, 
we run the modified Mix and Match mechanism 
over all possible  values of the $\lceil \log(m) \rceil$ random bits.
The collection of all resulted matchings is called {\em layer $0$ matchings}.
Note that the number of matchings for 
layer $0$ is at most $2^{\lceil \log(m) \rceil}\leq 2m$. 
We now 
describe $\lceil \log(m) \rceil$ steps to combine these matchings together, into a single matching. We note that each layer will halve the number 
of matchings, and so $\lceil \log(m) \rceil$ applications 
of the mechanism give a layer with a single matching and this matching is our output.
 After the $i$-th step we inductively 
construct the matchings of the $i+1$-th layer as follows. We decompose the matchings in the $i$-th layer into arbitrary pairs of matchings. We use the procedure  of Lemma \ref{lm:main} on every pair. Unlike the randomized version, here we always output the matching between 
$N_1$ and $N_2$ that has the largest number of edges.
Clearly, in each step, the number of matchings in the layer is halved. 
Thus, after $\lceil \log(m) \rceil$ steps we have exactly one matchings. 
This matching is the output of our mechanism.

This mechanism contains at most $2m$ runs of the modified Mix and Match mechanism and at most $2m$ combinations of such matchings. 
Both the modified Mix and Match mechanism and the combination procedure run in polynomial time. Thus, this mechanism is a polynomial time mechanism. 

Note that the average number of edges in the $0$-th layer is exactly equal to the expected social welfare in the modified Mix and Match mechanism. 
This follows because every labeling among the $2^{\lceil \log(m) \rceil}$
is equally likely.
In each step, we replace each pair with one of the matching obtained from Lemma \ref{lm:main}. Selecting the matching between 
$N_1$ and $N_2$  that contains the largest number of  edges,
combined with the second property of Lemma \ref{lm:main},
implies that  the average number of edges cannot decrease when we go 
to the next layer. 
Thus, the number of edges in the last-layer matching is at least that of the modified Mix and Match mechanism. Thus, this mechanism is a $2$-approximation mechanism.
We now discuss individual utilities.
The average utility of every agent in layer $0$ 
 exactly equals her expected utility in the modified Mix and Match mechanism. 
Using the first property of Lemma \ref{lm:main}, it is clear that for a fixed agent, the difference between its modified
Mix and Match utility, and the average of her utilities in the $1$-th layer is at most $1$. 
This holds true each time we go from one layer to the next.
Namely, for every agent, the difference in utility between the modified Mix and Match 
strategy and our strategy goes up by at most $1$. 
Thus, in the $i$-th layer the difference between the utility 
in our deterministic mechanism and 
the modified Mix and Match mechanism is at most $i$. 

We now inspect how much an agent can  
gain by deviating from the mechanism.
Since the modified Mix and Match mechanism is truthful, her utility in the modified Mix and Match mechanism does not increase with the new strategy.
The difference between her expected utility in the modified Mix and Match mechanism and our mechanism is at most $\lceil \log(m) \rceil$. 
A non truthful strategy can increase the utility by at most 
an additive factor of $2\lceil \log(m) \rceil$. \qed
\end{proofof}

\bibliography{kidney}
\bibliographystyle{plain}


\end{document}